%% file: for_arxiv.tex
\documentclass[onefignum,onetabnum]{siamonline171218}
\usepackage[T1]{fontenc}
\newtheorem{assumption}{Assumption}
\usepackage{tikz}
\usepackage{caption}
\usetikzlibrary{calc,angles,quotes,decorations.pathreplacing}
\newcommand*{\captionsource}[2]{%
  \caption[{#1}]{%
    #1%
    \\\hspace{\linewidth}%
    \textbf{Source:} #2%
  }%
}

\input{ex_shared}

\ifpdf
\hypersetup{
  pdftitle={The Omniscient yet Lazy Investor},
  pdfauthor={S.M.S. Halkiewicz}
}
\fi

\myexternaldocument{ex_supplement}

\begin{document}

\maketitle

%

\begin{abstract}
We formalize the paradox of an \emph{omniscient yet lazy investor}—a perfectly informed agent 
who trades infrequently due to execution or computational frictions.
Starting from a deterministic geometric construction, we derive a closed-form expected profit function linking trading frequency, execution cost, and path roughness. 
We prove existence and uniqueness of the optimal trading frequency and show that this 
optimum can be interpreted through the fractal dimension of the price path. 
A stochastic extension under fractional Brownian motion provides analytical expressions 
for the optimal interval and comparative statics with respect to the Hurst exponent. 
Empirical illustrations on equity data confirm the theoretical scaling behavior.
\end{abstract}



\section{Introduction}
\label{sec:intro}

Modern financial markets operate under a tension between the speed of information 
and the frictions of execution \cite{Kyle1985,Almgren2000,Obizhaeva2013}.  
Algorithmic traders, institutional portfolio managers, and even theoretical 
agents must continually decide not only \emph{what} to trade but also 
\emph{how frequently} to act.
In principle, an investor endowed with perfect foresight could exploit every 
profitable fluctuation of an asset's price path; 
in practice, each decision carries tangible costs --- transaction fees, 
market impact, computational latency, and cognitive or regulatory frictions 
\cite{Cartea2015,Sims2003,Gabaix2019}.  
The result is an optimization problem that balances 
omniscience with inertia: the investor knows everything, yet cannot act continuously.  

This paper formalizes that paradox through the stylized construct of 
the \emph{omniscient yet lazy investor}.  
The model postulates an agent with complete knowledge of the future price path, 
subject to additive execution costs and a cumulative penalty 
for the mere act of trading or recalculating --- a ``laziness cost'' representing 
bounded rationality, algorithmic latency, or decision fatigue 
\cite{Caplin2015,Fudenberg2018,Cartea2015}.  
The investor’s problem is to determine the number of trades (or, equivalently, 
the trading frequency) that maximizes expected total profit over a finite horizon.

Starting from a deterministic geometric setting, we derive a closed-form 
expression for the investor’s expected profit as a function of the number of 
trading intervals, the per-trade friction, and a parameter describing 
the roughness of the price path.  
This formulation extends the classical portfolio-rebalancing literature on 
transaction costs \cite{Magill1976,Davis1990,Shreve1994,Korn1998}, 
interpreting frequency choice as a discrete control variable rather than 
a boundary condition on wealth.  
The resulting trade-off is intuitive: increasing trading frequency 
magnifies exploitable price variation but also raises costs, 
and the total profit eventually decreases, producing a finite optimum.

A key insight of the model is its connection to fractal geometry.  
By relating the effective length of the price path to its scaling exponent, 
we interpret the optimal frequency through the \emph{fractal dimension} 
of the underlying trajectory.  
This view builds on the Fractal Market Hypothesis, originally proposed by Peters~\cite{Peters1994}, 
which attributes market stability to the coexistence of heterogeneous investment horizons.  
Empirical studies have shown that financial time series exhibit 
self-similarity and scale-invariant roughness \cite{Mandelbrot1967,Kristoufek2013,Liu2022}, 
motivating the explicit use of fractal measures in our analysis.  
In this interpretation, rougher (more irregular) price paths correspond to 
higher optimal trading frequencies.  

To ground the model in a stochastic environment, we extend the framework to 
price dynamics driven by fractional Brownian motion (fBM) with Hurst exponent 
$H \in (0,1)$ \cite{Mandelbrot1968,Biagini2008,Mishura2008}.  
Under the self-similarity property of fBM, 
the expected exploitable price increment scales as $\Delta^{H}$, 
and the deterministic profit formula generalizes naturally to a stochastic one.  
We obtain explicit expressions for the optimal rebalancing interval 
$\Delta^\star = [\,\bar s / (\kappa(1-H))\,]^{1/H}$, 
where $\bar s$ denotes the effective execution friction 
and $\kappa$ is a scaling constant determined by volatility and normalization.  
This result confirms the theoretical intuition: as the path becomes more fractal 
(smaller $H$), the investor should act more frequently 
\cite{Kakinaka2025,Bennedsen2021,Fukasawa2023}.  

Numerical and empirical examples illustrate the analytical results.  
Using equity data, we show that the observed profit function follows 
the predicted concave shape, and the empirically optimal frequency 
lies close to the theoretical one derived under the fBM approximation 
\cite{Bayraktar2004,FernandezMartinez2019,Verma2024,Alizade2025}.  
The analysis thus connects geometric properties of price paths with 
economically meaningful decisions about trading intensity, 
in the spirit of the fractal interpretation of market behavior 
\cite{Peters1994,Halkiewicz2024,Sornette2003,Mantegna1999}. 

The rest of the paper is organized as follows. 
In section~\ref{sec:litreview} we review the literature concerning fractal structure of capital markets and place our model within the existing frameworks.
Section~\ref{sec:model} introduces the model framework and notation.  
Section~\ref{sec:deterministic} derives the deterministic closed-form profit formula.  
Section~\ref{sec:optimization} studies the existence and properties of the 
optimal trading frequency.  
Section~\ref{sec:stochastic} develops the stochastic extension based on 
fractional Brownian motion and provides comparative statics.  
Section~\ref{sec:empirical} presents empirical and simulation evidence and Section~\ref{sec:conclusion} discusses implications for algorithmic trading, concludes and suggests pathways for future research.

\section{Literature Review}
\label{sec:litreview}

The present study builds upon two mature yet historically separate research traditions in mathematical finance:  
(1) the optimization of portfolio rebalancing under proportional transaction costs, and  
(2) the fractal modeling of market dynamics, including the Fractal Market Hypothesis (FMH).  
Bridging these frameworks provides a geometric and stochastic interpretation of trading frequency as an endogenous response to both frictions and price‐path roughness.

\subsection{Portfolio Optimization under Transaction Costs}

The study of optimal portfolio rebalancing in frictional markets originates from the seminal works of Davis and Norman~\cite{Davis1990} and Shreve and Soner~\cite{Shreve1994}.  
Their continuous‐time impulse‐control formulations established the existence of \emph{no‐trade regions}—intervals of inaction within which the marginal cost of trading exceeds the marginal benefit of rebalancing.  
These models formalized the balance between maintaining target allocations and minimizing cumulative transaction costs, and they provided a mathematical basis for the design of modern trading and execution strategies.  
Subsequent research, including Korn~\cite{Korn1998}, extended the analysis to stochastic volatility, portfolio insurance, and partial‐information settings, integrating dynamic programming and viscosity‐solution techniques.  
This literature underpins the first strand of the present work: trading frequency as an optimization variable constrained by proportional and cognitive frictions.

\paragraph{Fractal geometry and market structure.}
Fractal geometry was introduced to financial analysis by Mandelbrot~\cite{Mandelbrot1997},
whose pioneering studies on scaling laws in price series and heavy‐tailed distributions
challenged the Gaussian assumptions of the Efficient Market Hypothesis (EMH).  
Building on these ideas, Peters~\cite{Peters1994} proposed the Fractal Market Hypothesis (FMH), arguing that market stability depends on a heterogeneous spectrum of investment horizons,
with financial crises corresponding to periods dominated by short‐term trading behavior.  
FMH emphasises self-similarity, roughness, and scale invariance in market dynamics rather than perfect informational efficiency.

\subsection{Fractional Processes and Rough Volatility}

A rigorous stochastic foundation for fractal scaling was established by Mandelbrot and van Ness~\cite{Mandelbrot1968}, who defined fractional Brownian motion (fBM) and fractional noise as models with self‐similar increments and memory parameterized by the Hurst exponent \(H\in(0,1)\).  
The Hurst exponent determines both the persistence of increments and the Hausdorff (fractal) dimension \(D = 2 - H\) of the sample path.  
Applications of fBM to asset prices and volatility were advanced by Comte and Renault~\cite{Comte1998}, Wyss~\cite{Wyss2000} (after Kim et al.~\cite{Kim2021} and Zhang et al.~\cite{Zhang2024}), and Gatheral, Jaisson, and Rosenbaum~\cite{Gatheral2018}, the latter providing strong empirical evidence that volatility is \textit{rough}, with typical \(H\) near 0.1.  
This line of work firmly connected fractal geometry with stochastic modeling, showing that memory and roughness jointly govern volatility clustering and scaling behavior.

\subsection{Modern Developments in Fractal Financial Modeling}

Recent research has revitalized fractal approaches using both empirical and analytical tools.  
Wu et~al.~\cite{Wu2021} introduced \emph{fractal statistical measures}—fractal expectation and variance—to construct portfolio selection models under power‐law tails, yielding closed‐form weights that outperform traditional mean–variance optimization.  
Kakinaka et~al.~\cite{Kakinaka2025} studied \emph{fractal portfolio strategies} in which investor preferences over temporal scales influence performance and risk.  
El‐Nabulsi and Anukool~\cite{El-Nabulsi2025} extended this perspective to markets defined in fractional dimensions, deriving qualitative properties of asset dynamics within noninteger geometric spaces.  

Parallel advances in measurement techniques have enhanced empirical precision.  
Bayraktar, Poor, and Sircar~\cite{Bayraktar2004} estimated the fractal dimension of the S\&P~500 index via wavelet analysis, linking declining Hurst exponents to increasing market efficiency.  
Verma and Kumar~\cite{Verma2024} analyzed post–merger financial performance using fractal interpolation and box dimension metrics,  
while Alizade et~al.~\cite{Alizade2025} modeled market turbulence through Laplace–Mittag–Leffler distributions, capturing heavy tails and memory effects beyond classical Lévy frameworks.  
The mathematical foundations of fractal dimension estimation in applied finance are comprehensively presented by Fernández–Martínez et~al.~\cite{FernandezMartinez2019}, whose work consolidates analytical and numerical techniques for quantifying complexity in time‐series data. Halkiewicz~\cite{Halkiewicz2024} provided a conceptual synthesis of market graphs as fractals. 
This interpretation aligns with the FMH and emphasizes that the complexity of market dynamics increases as the observation interval shortens—a principle central to the present paper’s formulation of frequency‐dependent profitability.

\subsection{Positioning of the Present Study}

The contribution of this work is to unify the frictional and fractal paradigms within a single analytical framework.  
While classical transaction‐cost models quantify the cost–benefit trade‐off of frequent trading, and fractal market models describe the geometry of price fluctuations, the two are rarely connected formally.  
Our model links them by interpreting the investor’s optimal trading frequency as a function of the fractal dimension (or equivalently, the Hurst exponent) of the underlying price path.  
This geometric–economic synthesis provides a closed‐form solution to the “omniscient yet lazy investor” problem and delivers a quantitative manifestation of the Fractal Market Hypothesis:  
as the market becomes rougher and more fractal, optimal trading frequency increases, reflecting the higher information content per unit time.

\section{Model Framework}
\label{sec:model}

We consider an \emph{omniscient investor}. 
Apart from knowing the answers to all metaphysical questions about the universe and everything in it, 
he naturally possesses perfect foresight regarding future price movements as well. 
He is also incurably greedy: 
\begin{itemize}
    \item having already solved the Millennium Problems,
    \item discovered every lost treasure,
    \item and won so many games of poker that every casino in Las Vegas has him blacklisted,
\end{itemize}

he still desires more and more money. Determined to turn his omniscience into yet another source of amusement and wealth, he decides to participate in the financial markets. 

Unfortunately, omniscience does not preclude indolence. 
Our investor is also profoundly lazy. Each trade, however trivial in its consequence, requires effort, attention, and perhaps the faint movement of a finger—actions he finds increasingly tiresome. 
He faces proportional execution frictions whenever he deigns to act, and cumulative decision-making costs whenever he considers doing so. 
Such frictions are pervasive in both theoretical and empirical market models \cite{Magill1976,Davis1990,Shreve1994,Kallsen2015}. 
His problem, therefore, is quintessentially human despite his divine insight: 
given total knowledge and total apathy, how often should he trade within a fixed horizon in order to maximize his expected profit?

\subsection{Discrete trading grid}

Let $T>0$ denote the total investment horizon.  
We partition $[0,T]$ into $n$ subintervals of equal length
$\Delta=T/n$, corresponding to $n=2^{m}$ trading periods indexed by
$i=0,1,\dots,n$.  
Denote by $c_i$ the asset price at time $t_i=i\Delta$.
The increment between successive observation points is
\begin{equation}
\Delta c_i = c_{i+1}-c_i.
\label{eq:increment}
\end{equation}
Such discrete rebalancing grids are standard in the literature on optimal portfolio revision and dynamic trading \cite{Magill1976,Korn1998,Cartea2015}.

Because the investor is \emph{omniscient}, each $\Delta c_i$ is known in advance.
However, the investor incurs two forms of friction when acting on this knowledge:
\begin{enumerate}
  \item a per‐trade execution cost or spread $\bar s\ge0$,
  \item an additive ``laziness cost'' $l_i\ge0$ representing 
        the cognitive, computational, or opportunity cost
        of taking a decision at time $t_i$.
\end{enumerate}
The first term represents proportional costs studied in transaction‐cost models such as those of Davis and Norman or Shreve and Soner \cite{Davis1990,Shreve1994}, 
while the second reflects decision‐making or latency penalties 
that parallel the cognitive constraints emphasized in rational‐inattention theory \cite{Sims2003,Gabaix2019,Caplin2015}.  
The total laziness cost accumulated over the horizon is denoted
\begin{equation}
L=\sum_{i=1}^{n} l_i.
\label{eq:laziness}
\end{equation}
This cumulative term generalizes beyond monetary costs, encompassing computational energy expenditure or machine‐learning inference delays in automated systems \cite{Cartea2015}.

\subsection{Profit identity}

For an omniscient trader who always takes the correct side of the market,
the gross gain over period $i$ equals the absolute price change
$|\Delta c_i|$.
After subtracting costs, the realized profit over all periods is
\begin{equation}
R = \sum_{i=1}^{n} \big(|\Delta c_i| - \bar s - l_i\big).
\label{eq:Rsum}
\end{equation}
Defining the mean exploitable return per trade as
\(
\bar r = \tfrac{1}{n}\sum_{i=1}^{n}|\Delta c_i|,
\)
we can rewrite \eqref{eq:Rsum} as
\begin{equation}
R = n(\bar r - \bar s) - L.
\label{eq:Rcompact}
\end{equation}
Expression~\eqref{eq:Rcompact} provides the fundamental relationship 
between trading frequency and total profit.  
As the number of trades $n$ increases, the exploitable mean return 
$\bar r$ typically rises because smaller intervals reveal additional
micro‐movements of the price path, a phenomenon consistent with the 
scaling laws documented in high‐frequency data \cite{Bouchaud2000,Cont2001,Calvet2002}.  
At the same time, both the proportional cost $\bar s$ and the total laziness cost $L$
reduce net profitability, in line with the classic execution–latency trade‐off in optimal trading theory \cite{Almgren2000,Obizhaeva2013,Kyle1985}.

\subsection{Decision variable and optimization problem}

The investor chooses the number of trading intervals $n$, or equivalently
the dyadic level $m$ such that $n=2^{m}$, to maximize $R$:
\begin{equation}
\max_{m\in\mathbb{N}} \; R_m := 2^{m}(\bar r_m - \bar s) - L_m.
\label{eq:optproblem}
\end{equation}
Here $\bar r_m$ and $L_m$ denote, respectively, the average exploitable 
return per trade and the cumulative laziness cost when the price path
is observed at resolution $m$.  
The trade‐off mirrors impulse‐control models in which agents balance 
costly rebalancing with expected drift \cite{Constantinides1986,Taksar1988,Kallsen2015}, 
but the present setting replaces the stochastic control boundary 
with a discrete frequency variable capturing self-similar resolution.

\subsection{Interpretation of costs}

The term $\bar s$ encompasses all proportional costs that scale 
linearly with the number of trades, 
including bid–ask spreads, slippage, and market‐impact fees \cite{Almgren2000,Cartea2015,Obizhaeva2013}.  
The term $L_m$ represents non‐linear or sublinear costs of action:
for example, human cognitive effort \cite{Fudenberg2018}, 
machine‐learning inference latency \cite{Cartea2015}, 
or computational resource usage.  
Allowing $L_m$ to grow with $m$ captures the intuitive idea that higher 
trading frequency requires disproportionately greater informational and 
technological capacity \cite{Gabaix2019,Sims2003}.
Such latency–dependent frictions have been studied extensively in HFT models, 
where execution speed, information flow, and order–book resilience jointly determine profitability 
\cite{Cartea2014,Jarrow2015}.

\subsection{Fractal scaling motivation}

Empirically, as the sampling interval $\Delta$ decreases,
the measured variation of a financial price series increases 
in a manner reminiscent of fractal scaling \cite{Mandelbrot1967,Mandelbrot1997,Calvet2002}.  
This observation motivates modeling $\bar r_m$
as a function of the effective \emph{roughness} of the price path,
a property long noted in the context of fractal market geometry \cite{Peters1994,Gatheral2018book}.  
In Section~\ref{sec:deterministic} we formalize this dependence and derive a closed‐form expression for $R_m$
based on a geometric construction involving a scaling parameter~$W$
that quantifies the fractal complexity of the trajectory \cite{Feder1988,Lux2003}.

\section{Deterministic Fractal Derivation}
\label{sec:deterministic}

This section provides a geometric scaling argument that links the exploitable
per–trade move to the sampling resolution and a roughness parameter.
Combined with~\eqref{eq:Rcompact}, it yields a closed-form expression for
$R_m$ as a function of the dyadic level~$m$.

\subsection{Triangle construction and scaling postulate}

The idea that apparent path length depends on observation scale is central in
fractal geometry \cite{Mandelbrot1967,Feder1988}.  
We apply it here in the simplest possible way.

Fix the horizon $T>0$ and a dyadic resolution $m\in\mathbb{N}$ with
$n=2^m$ subintervals and step $\Delta=T/2^m$.  
Over one subinterval $[t_i,t_{i+1}]$, we represent the effective local
displacement by a right triangle with horizontal leg $\Delta$, vertical leg
$\bar h_m\ge0$ (the mean exploitable move), and an auxiliary microstructure
scale $W^m c_0$ that captures residual oscillations at resolution $m$
\cite{Peters1994,Mandelbrot1997}.  

\begin{assumption}[Geometric scaling]
\label{ass:scaling}
There exist constants $W>0$ and $c_0>0$ such that for each dyadic level $m$,
the components of the local displacement satisfy the Pythagorean relation
\begin{equation}
\Big(\tfrac{T}{2^m}\Big)^{\!2}
\;=\;
\bar h_m^{\,2} \;+\; W^{2m} c_0^{\,2}.
\label{eq:pyth}
\end{equation}
\end{assumption}

Identity~\eqref{eq:pyth} is purely geometric: the chord length per subinterval
is fixed by the sampling step, while the vertical excursion and the
microstructure term trade off as resolution changes.  
This representation mirrors classical self-affine constructions in fractal
curves, such as those used to define the coastline paradox or Brownian paths
\cite{Feder1988,Mandelbrot1997,Lux2003}.  

\begin{figure}[h]
\centering
\begin{tikzpicture}[scale=2.1, line cap=round, line join=round]

\coordinate (A1) at (0,0);
\coordinate (B1) at (1.9,0);
\coordinate (C1) at (1.9,1.25);

\coordinate (A2) at (3.30,0);
\coordinate (B2) at (5.10,0);
\coordinate (C2) at (5.10,1.05);

\draw[thick] (A1)--(B1)--(C1)--cycle;
\draw[thick] (A2)--(B2)--(C2)--cycle;

\draw ($(B1)+(-0.14,0)$) -- ($(B1)+(-0.14,0.14)$) -- ($(B1)+(0,0.14)$);
\draw ($(B2)+(-0.14,0)$) -- ($(B2)+(-0.14,0.14)$) -- ($(B2)+(0,0.14)$);

\draw[thin] (A1) ++(0:0.22) arc (0:33:0.22);
\draw[thin] (A2) ++(0:0.22) arc (0:30:0.22);

\node[scale=0.95] at ($(A1)!0.5!(C1)+(0,1.05)$) {singular $i$th period};
\node[scale=0.95] at ($(A2)!0.5!(C2)+(0,0.95)$) {averaged period};

\node[below=4pt] at ($(A1)!0.5!(B1)$) {$\Delta=\dfrac{T}{2^m}$};
\node[below=4pt] at ($(A2)!0.5!(B2)$) {$\Delta=\dfrac{T}{2^m}$};

\draw[decorate,decoration={brace,amplitude=3pt}]
      (B1) -- (C1) node[midway,xshift=10pt] {$h_{ni}$};
\node[sloped,above=2pt,fill=white,inner sep=1pt]
      at ($(A1)!0.58!(C1)$) {$c_i$};

\draw[decorate,decoration={brace,amplitude=3pt}]
      (B2) -- (C2) node[midway,xshift=10pt] {$\bar h_m$};
\node[sloped,above=2pt,fill=white,inner sep=1pt]
      at ($(A2)!0.58!(C2)$) {$W^m c_0$};

\end{tikzpicture}

\caption{Right-triangle construction per subinterval at resolution $m$.  
Both triangles share base $\Delta=T/2^m$ and satisfy
$\left(\tfrac{T}{2^m}\right)^2=\bar h_m^2+W^{2m}c_0^2$.  
The relation models the scale-dependent roughness of price increments 
observed empirically in fractal market studies \cite{Calvet2002,Bouchaud2000}.}
\label{fig:triangle}
\end{figure}
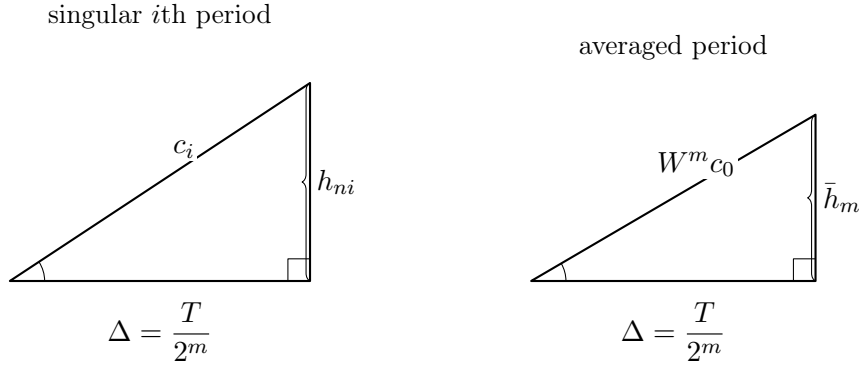

Solving~\eqref{eq:pyth} gives the exploitable mean move
\begin{equation}
\bar h_m
\;=\;
\sqrt{\frac{T^2}{4^m} - W^{2m} c_0^{\,2}},
\qquad
\text{feasible iff }\;\frac{T}{2^m} > W^m c_0 .
\label{eq:hbar}
\end{equation}
The feasibility restriction ensures that at high enough resolution the
microstructure noise dominates, echoing the breakdown of scaling observed in
empirical data when market microstructure effects appear
\cite{Cont2001,Calvet2002}.

\subsection{Closed form for the profit function}

Substituting $\bar h_m$ from~\eqref{eq:hbar} into the identity
$R_m = n(\bar r_m-\bar s) - L_m$ with $\bar r_m\equiv \bar h_m$ and $n=2^m$
yields the central expression for total profit at resolution $m$.

\begin{proposition}[Closed-form profit at dyadic level $m$]
\label{prop:R-closed}
Under Assumption~\ref{ass:scaling}, for every $m$ satisfying
$\frac{T}{2^m} > W^m c_0$,
\begin{equation}
R_m
\;=\;
2^m\!\left(
\sqrt{\frac{T^2}{4^m} - W^{2m} c_0^{\,2}}
\;-\; \bar s
\right)
\;-\; L_m .
\label{eq:R-closed}
\end{equation}
\end{proposition}

\begin{proof}
By~\eqref{eq:pyth}–\eqref{eq:hbar}, the exploitable mean move per trade at
level $m$ equals $\bar r_m=\bar h_m$ provided the square root is real, i.e.,
$\frac{T}{2^m} > W^m c_0$.
Then $R_m = 2^m(\bar r_m-\bar s)-L_m$ gives~\eqref{eq:R-closed}.
\end{proof}

\begin{remark}[Feasibility region and qualitative behavior]
\label{rem:feasible}
The feasibility condition $\frac{T}{2^m} > W^m c_0$ defines an upper bound
$m<m_{\max}$, where $m_{\max}$ is the largest integer with
$2^{-m} > W^m c_0/T$.
As $m$ increases, the chord $\tfrac{T}{2^m}$ shrinks while the microstructure
term $W^{m}c_0$ scales geometrically; hence the radicand in
\eqref{eq:hbar} decreases and eventually becomes negative, at which point the
model predicts no exploitable move at that resolution.  
This upper limit parallels the practical observation that returns lose
scaling coherence beyond microsecond horizons in high-frequency markets
\cite{Bouchaud2000,Cartea2015}.
\end{remark}

\subsection{Comparative statics at the deterministic level}

Write
\[
\Phi(m;T,W,c_0)
\;:=\;
\sqrt{\tfrac{T^2}{4^m} - W^{2m} c_0^{\,2}},
\quad\text{so}\quad
R_m = 2^m\big(\Phi(m;T,W,c_0)-\bar s\big)-L_m .
\]
Within the feasible set:
\begin{itemize}
  \item $R_m$ decreases linearly in $\bar s$ and in $L_m$, 
        reflecting the standard frictional mechanism of Davis and Norman \cite{Davis1990}.
  \item $R_m$ declines with $W$ and $c_0$ since larger microstructure
        intensity reduces net exploitable motion, consistent with empirical
        microstructure estimates \cite{Obizhaeva2013,Cont2001}.
  \item There is a discrete trade-off in $m$: the multiplier $2^m$ favors
        finer sampling, while $\Phi(m;\cdot)$ typically shrinks with $m$ and
        feasibility eventually fails (Remark~\ref{rem:feasible}).
\end{itemize}
This deterministic structure prepares the ground for the stochastic
generalization in Section~\ref{sec:stochastic}, where 
$\Phi(m;\cdot)$ will emerge from the self-similar scaling law of fractional Brownian motion \cite{Biagini2008, Duncan2000}.


\section{Optimization of Trading Frequency}
\label{sec:optimization}

We now study the choice of resolution $m$ (equivalently, the number of
trades $n=2^m$) that maximizes the total profit $R_m$ given in
\eqref{eq:R-closed}.  Denote the feasible index set by
\[
\mathcal{M}\;:=\;\{\,m\in\mathbb{N}: \tfrac{T}{2^m} > W^m c_0\,\}\;=\;\{0,1,\dots,m_{\max}\}.
\]

\subsection{Existence (finite feasible set)}

The discrete maximization problem resembles the impulse–control problems of
Davis and Norman~\cite{Davis1990} and Shreve and Soner~\cite{Shreve1994}, 
where optimal rebalancing is characterized by no–trade regions rather than
continuous adjustment.  
Here, the feasible set $\mathcal{M}$ is finite, which immediately yields the
existence of a maximizer.

\begin{proposition}[Existence of a maximizer]
\label{prop:existence-finite}
For fixed parameters $(T,W,c_0,\bar s)$ and any nonnegative cost sequence
$\{L_m\}_{m\in\mathcal{M}}$, the maximization problem
$\max_{m\in\mathcal{M}} R_m$ admits at least one solution
$m^\star\in\mathcal{M}$.
\end{proposition}

\begin{proof}
$\mathcal{M}$ is finite and $R_m$ is real-valued on $\mathcal{M}$ by
Proposition~\ref{prop:R-closed}, hence the maximum is attained.
\end{proof}

\subsection{Marginal characterization via forward differences}

Discrete changes in $m$ play the same role here as marginal time adjustments in
continuous trading models \cite{Constantinides1986,Taksar1988,Kallsen2015}.
Define the forward difference of total profit
\[
\Delta R_m \;:=\; R_{m+1}-R_m,
\qquad m=0,1,\dots,m_{\max}-1.
\]
Using $R_m = 2^m(\Phi_m-\bar s)-L_m$ with
$\Phi_m:=\sqrt{\tfrac{T^2}{4^m}-W^{2m}c_0^{\,2}}$, a direct calculation gives
\begin{equation}
\Delta R_m
=
2^m\!\Big(2\Phi_{m+1}-\Phi_m-\bar s\Big)
\;-\;\Delta L_m,
\qquad
\Delta L_m:= L_{m+1}-L_m\ge 0.
\label{eq:DeltaR}
\end{equation}

\begin{theorem}[Marginal stopping rule]
\label{thm:stopping}
Suppose $L_m$ is nondecreasing and
$2\Phi_{m+1}-\Phi_m$ is nonincreasing in $m$
(\emph{diminishing marginal exploitable move}).  
Then $R_m$ is unimodal on $\mathcal{M}$, and any maximizer $m^\star$ is
characterized by the smallest index for which $\Delta R_m\le 0$:
\begin{equation}
m^\star
\;=\;
\min\big\{\,m\in\mathcal{M}:\ \Delta R_m\le 0\,\big\}.
\label{eq:threshold}
\end{equation}
In particular, if $\Delta R_m>0$ for all $m<m_{\max}$ then $m^\star=m_{\max}$,
while if $\Delta R_0\le 0$ then $m^\star=0$.
\end{theorem}

\begin{proof}
The argument parallels the discrete concavity reasoning of
Proposition~\ref{prop:unique} and impulse‐control logic
\cite{Davis1990,Constantinides1986}.  
Monotonicity of $L_m$ and $2\Phi_{m+1}-\Phi_m$ implies that the sequence of
differences $\{\Delta R_m\}$ is nonincreasing, so it can cross zero at most
once.  Consequently, $\{R_m\}$ is unimodal and
\eqref{eq:threshold} identifies the first nonpositive increment as the optimum.
\end{proof}

\begin{remark}[Economic interpretation]
\label{rem:eco}
The term $2^m(2\Phi_{m+1}-\Phi_m)$ represents the \emph{gross marginal benefit}
from refining the trading grid from $m$ to $m{+}1$, i.e., doubling the number
of trades.  The terms $2^m\bar s$ and $\Delta L_m$ represent, respectively, the
\emph{marginal execution friction} and the \emph{incremental cognitive or
computational cost}.  The stopping rule $\Delta R_m\le0$ thus formalizes the
principle known from dynamic–trading literature
\cite{Almgren2000,Cartea2015,Obizhaeva2013}: increase trading frequency only
until the marginal gain equals the incremental total cost.
\end{remark}

\begin{figure}[th!]
    \centering
    \includegraphics[width=\textwidth]{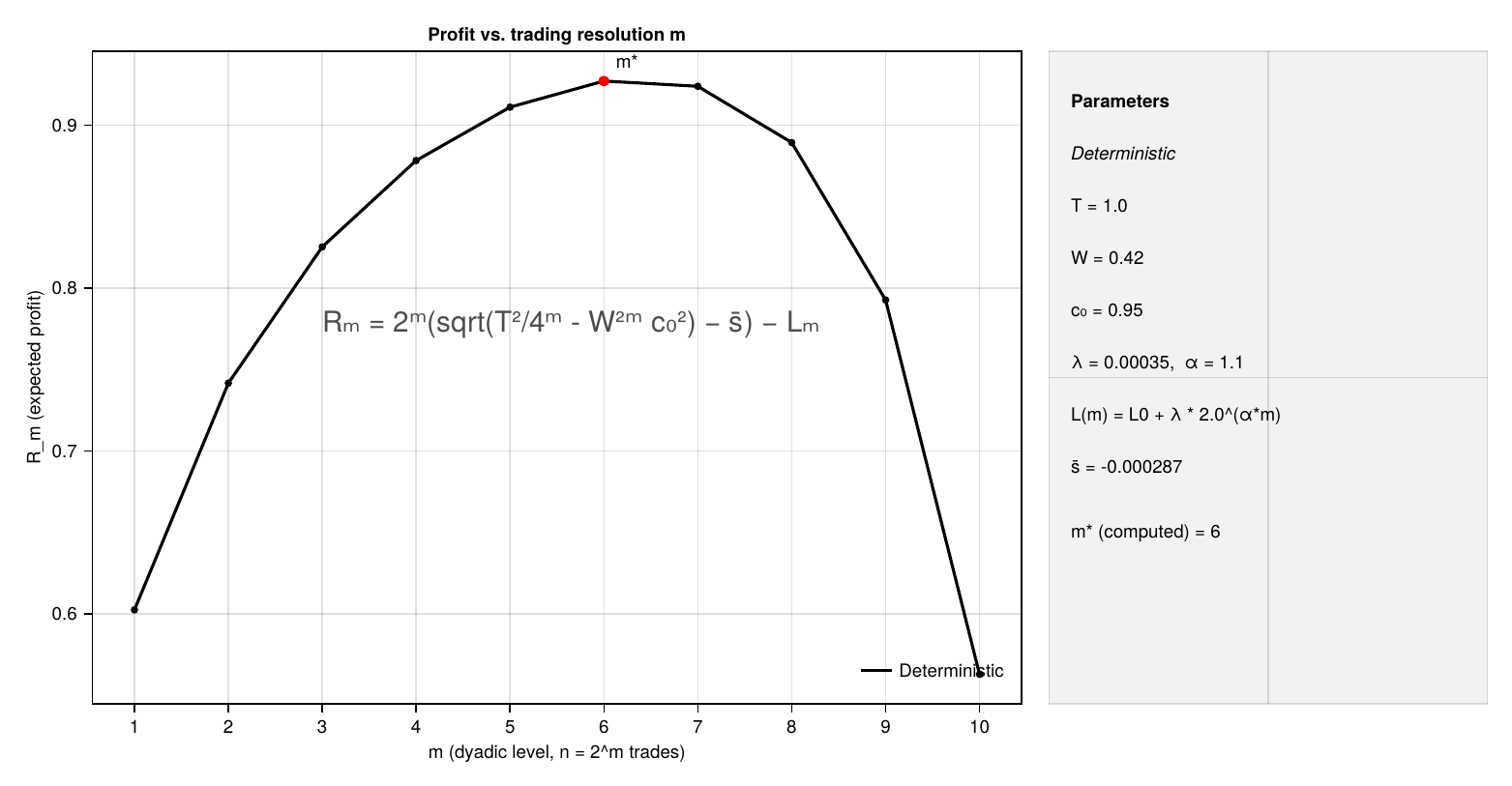}
    \captionsource{Expected profit $R_m$ as a function of trading resolution $m$ 
    in the deterministic framework of Section~\ref{sec:deterministic}.
    The curve illustrates the trade--off between the exploitable 
    substructure gain (first term) and proportional plus cognitive costs 
    ($2^m\bar s + L(m)$). The interior maximum corresponds to the optimal 
    trading interval $\Delta^\star = T/2^{m^\star}$.}{Own calculations performed in Julia}
    \label{fig:Rm_vs_m}
\end{figure}

The shape of $R_m$ captures the essential logic of the omniscient investor’s dilemma.
At coarse resolutions (small~$m$), trading is too infrequent to exploit
micro–fluctuations in the price path, resembling under‐trading equilibria
observed in bounded–rationality models \cite{Sims2003,Gabaix2019}.  
As the resolution increases, the exploitable deterministic variation grows and
so does attainable profit, as in high–frequency execution models
\cite{Almgren2000,Cartea2015}.  
Beyond the interior optimum~$m^\star$, however, the rapidly compounding
execution frictions and cognitive costs dominate, causing total expected profit
to decline—a pattern consistent with concave profit functions in transaction–
cost theory \cite{Davis1990,Shreve1994,Korn1998}.  
The resulting hump-shaped profile visually represents the analytical
first–order condition and interprets $\Delta^\star = T/2^{m^\star}$ as the
frequency at which greed and laziness balance.

\subsection{Sufficient conditions for uniqueness}

Let $A_m:=2^m\Phi_m$ and write $R_m=A_m-2^m\bar s-L_m$.

\begin{assumption}[Regularity]
\label{ass:regularity}
(i) $\{A_m\}_{m\in\mathcal{M}}$ is strictly concave in the discrete sense:
$\Delta^2 A_m:=A_{m+2}-2A_{m+1}+A_m < 0$ for all $m$ with
$m{+}2\in\mathcal{M}$. \\
(ii) $L_m$ is convex and nondecreasing on $\mathcal{M}$.
\end{assumption}

These curvature conditions are analogous to standard assumptions guaranteeing
unique controls in stochastic optimization and dynamic programming
\cite{Taksar1988,Kallsen2015,Liu2013}.

\begin{proposition}[Strict unimodality and uniqueness]
\label{prop:unique}
Under Assumption~\ref{ass:regularity}, the sequence $\{R_m\}$ is strictly
unimodal and the maximizer $m^\star$ is unique.
\end{proposition}

\begin{proof}[Proof sketch]
Discrete concavity of $A_m$ and convexity of $2^m\bar s+L_m$ imply that
$\Delta R_m$ is strictly decreasing in $m$, hence it changes sign at most once
and the maximizer is unique \cite{Kallsen2015,Liu2013}.
\end{proof}

Full proof of Theorem~\ref{prop:unique} can be found in appendix~\ref{app:proofs}.

\begin{corollary}[Power–law laziness cost]
\label{cor:power}
If $L_m=\lambda\,2^{\alpha m}$ with $\lambda\ge 0$ and $\alpha\ge 1$
(\emph{linear or superlinear computational/latency growth}),
then Assumption~\ref{ass:regularity}(ii) holds and the maximizer is unique
provided $A_m$ is discretely concave.  
This specification is consistent with the convex energy or latency cost models
used in high–frequency algorithmic trading \cite{Cartea2015}.  
In particular,
\[
\Delta R_m
=
2^m\!\Big(2\Phi_{m+1}-\Phi_m-\bar s\Big)
\;-\;\lambda(2^{\alpha(m+1)}-2^{\alpha m}),
\]
and the stopping rule~\eqref{eq:threshold} applies.
\end{corollary}

\subsection{Bounds and comparative statics}

Within the feasible region,
\[
\frac{\partial R_m}{\partial \bar s}=-2^m \;<\;0,\qquad
\frac{\partial R_m}{\partial c_0}
=2^m\,\frac{-W^{2m}c_0}{\Phi_m} \;<\;0,\qquad
\frac{\partial R_m}{\partial W}
=2^m\,\frac{-W^{2m-1} c_0^{\,2}}{\Phi_m} \;<\;0.
\]
Hence the optimal index $m^\star$ is (weakly) decreasing in each of $\bar s$,
$W$, and $c_0$, and (weakly) decreasing in any parameter that increases $L_m$
pointwise.  
In words: higher frictions, rougher effective microstructure, or
larger computation costs shift the optimizer toward \emph{less} frequent trading,
a result consistent with comparative statics in continuous–time 
transaction–cost equilibria \cite{Shreve1994,Korn1998,Liu2013}.

\subsection{Practical algorithm (discrete argmax)}

Given $(T,W,c_0,\bar s)$ and a specification for $L_m$:

\begin{enumerate}
\item Compute $m_{\max}$ from feasibility (Remark~\ref{rem:feasible}).
\item For $m=0,1,\dots,m_{\max}-1$, evaluate $\Delta R_m$ via~\eqref{eq:DeltaR}.
\item If some $\Delta R_m\le 0$, set $m^\star = \min\{m:\Delta R_m\le 0\}$.
Otherwise set $m^\star=m_{\max}$.
\item (Optional) Verify uniqueness by checking $\Delta R_{m^\star-1}>0$ and
$\Delta R_{m^\star}\le 0$.
\end{enumerate}

This procedure is $O(m_{\max})$ and numerically stable because it avoids
subtracting large close numbers in $R_{m+1}-R_m$; all terms remain positive and
well scaled as long as feasibility is enforced \cite{Almgren2000,Cartea2015}.


\section{Stochastic Extension via Fractional Brownian Motion}
\label{sec:stochastic}

The deterministic framework of Section~\ref{sec:deterministic}
can be interpreted as a scaling law for the exploitable price increments
as the observation interval $\Delta$ varies.
We now embed this scaling in a stochastic process with well-defined
self-similarity and fractal properties --- fractional Brownian motion (fBM) \cite{Mandelbrot1968,Biagini2008,Mishura2008}, very widely used in financial applications \cite{Guasoni2019}.
This provides a probabilistic foundation for the model and allows explicit comparative statics with respect to the roughness of the price path, in line with modern evidence on rough volatility \cite{Bennedsen2021, Bayer2023ch2}.

\subsection{Fractional Brownian motion model}

Let $\{B^H_t\}_{t\ge 0}$ denote a fractional Brownian motion with
Hurst index $H\in(0,1)$, mean zero, and covariance
\[
\mathbb{E}[B^H_t B^H_s]
= \tfrac{1}{2}\Big(t^{2H}+s^{2H}-|t-s|^{2H}\Big).
\]
The process is $H$-self-similar and has stationary increments:
for all $\lambda>0$,
$B^H_{\lambda t}-B^H_{\lambda s}\stackrel{d}{=}\lambda^H(B^H_t-B^H_s)$ \cite{Mandelbrot1968}.
For $H=\tfrac12$ the process reduces to standard Brownian motion,
while for $H<\tfrac12$ the increments are negatively correlated
and the sample paths are rougher (fractal dimension $D=2-H$) \cite{Biagini2008,Mishura2008}.

We model the log-price process $X_t$ as
\[
X_t = \mu t + \sigma B^H_t,
\]
where $\mu$ and $\sigma$ are constant drift and volatility parameters.
The absolute log-return over an interval of length $\Delta$
is $|\Delta X_i| = \sigma|B^H_{t_{i+1}}-B^H_{t_i}|$,
whose expected value satisfies
\begin{equation}
\mathbb{E}|\Delta X_i|
  = \kappa(H,\sigma)\,\Delta^{H},
  \qquad
  \kappa(H,\sigma)
  := \sigma\,\mathbb{E}|B^H_1|
  = \sigma\,C,
\label{eq:EH}
\end{equation}
with $C=\sqrt{2/\pi}$ because $B^H_1\sim\mathcal{N}(0,1)$ for all $H$ \cite{Biagini2008,Mishura2008}.
Equation~\eqref{eq:EH} formally connects the exploitable mean move per trade
to the fractal scaling exponent $H$ and is consistent with empirical estimates of $H$ from high-frequency data \cite{Bayraktar2004}.

\subsection{Expected profit function}

For an omniscient trader who takes the correct direction of each move,
the expected gross gain over the horizon $T$ when trading every $\Delta$
time units (i.e., $n=T/\Delta$ trades) is
\[
\mathbb{E}[{\rm Gain}]
  = \frac{T}{\Delta}\, \kappa(H,\sigma)\,\Delta^{H}
  = \kappa T \Delta^{H-1},
\]
where we write $\kappa=\kappa(H,\sigma)$ for brevity.
Subtracting proportional execution frictions $\bar s$ and
aggregate laziness cost $L$ yields the total expected profit.

\begin{proposition}[Expected profit under fBM scaling]
\label{prop:fBM}
For $\Delta=T/2^m$,
\begin{equation}
R(\Delta)
  = \kappa T \Delta^{H-1}
  - \frac{T\bar s}{\Delta}
  - L.
\label{eq:R_fbm}
\end{equation}
\end{proposition}

\begin{proof}
The first term represents the expected exploitable return from $n=T/\Delta$ trades via \eqref{eq:EH}; 
the second aggregates per-trade frictions; the third is the total laziness cost.
\end{proof}

Expression~\eqref{eq:R_fbm} mirrors the deterministic formula
\eqref{eq:R-closed}, with the role of geometric roughness now played by the stochastic scaling exponent $H$ \cite{Mandelbrot1968}.

This scaling property also underlies recent theoretical analyses of high-frequency trading under fractional Brownian motion dynamics.  
In particular, Guasoni, Mishura, and R{\'a}sonyi~\cite{Guasoni2021} 
show that in the high-frequency limit, the conditionally expected increments of fBM converge to a white noise. Their results demonstrate that trading costs endogenously impose a finite optimal frequency, consistent with the frictional bound derived in our model.

\subsection{Optimal rebalancing interval}

Treating $\Delta$ as a continuous decision variable,
we maximize \eqref{eq:R_fbm} over $\Delta>0$.
Differentiating with respect to $\Delta$ (and omitting the constant $L$)
yields the first-order condition
\[
\frac{dR}{d\Delta}
  = \kappa T (H-1)\Delta^{H-2}
    + T\bar s\,\Delta^{-2}
  = 0.
\]
Solving for $\Delta$ gives the unique interior optimizer.

\begin{theorem}[Optimal trading interval]
\label{thm:opt}
Under $\bar s>0$ and $\kappa>0$,
the profit function \eqref{eq:R_fbm} has a unique maximizer
\begin{equation}
\Delta^\star
  = \Big(\frac{\bar s}{\kappa(1-H)}\Big)^{1/H}.
\label{eq:Delta_star}
\end{equation}
The corresponding optimal number of trades is
$n^\star = T / \Delta^\star$.
\end{theorem}

\begin{proof}
For $H\in(0,1)$, $R(\Delta)$ is strictly concave in $\Delta^H$ and the 
first-order condition admits a single positive root, yielding \eqref{eq:Delta_star}.
\end{proof}

\begin{corollary}[Fractal comparative statics]
\label{cor:fbmstatics}
Let $D = 2 - H$ denote the Hausdorff dimension of the fBM sample paths.
Then
\[
\frac{\partial \Delta^\star}{\partial D} < 0,
\quad
\frac{\partial \Delta^\star}{\partial \bar s} > 0,
\quad
\frac{\partial \Delta^\star}{\partial \kappa} < 0.
\]
Hence more fractal (rougher) price paths - larger $D$ or smaller $H$ - imply
a smaller optimal interval $\Delta^\star$, i.e., higher optimal trading frequency,
which is consistent with rough-volatility evidence and the FMH perspective on short-horizon dominance during turbulence \cite{Gatheral2018book,Kristoufek2013,Peters1994}.
\end{corollary}

Full proofs of Theorem~\ref{thm:opt} and Corollary~\ref{cor:fbmstatics} can be found in appendix~\ref{app:proofs}.

Figure~\ref{fig:Rm_vs_m_fbm} illustrates the comparative statics implied by
\eqref{eq:R_fbm} and Theorem~\ref{thm:opt}: rougher paths (smaller $H$) push the
optimum toward higher trading frequency, whereas the convex laziness penalty
$L(m)$ governs the sharpness of the decline beyond $m^\star$.

\subsection{Simulation under fractional Brownian motion}

To verify the robustness of the analytical relation
\eqref{eq:Delta_star}, we conducted numerical simulations of the model
under fractional Brownian motion (fBM) price dynamics. Synthetic log--price paths of unit length $T=1$ were generated for three representative Hurst exponents $H \in \{0.40,\,0.60,\,0.80\}$ using both the Cholesky decomposition \cite[Ch. 17]{Borgers2022} and the improved
Davies--Harte circulant-embedding method
\cite{Dieker2003, Craigmile2003}.
For each $H$, we evaluated the expected profit function
\eqref{eq:R_fbm} on a dyadic grid
$\Delta = T/2^{m}$ for $m=1,\dots,12$,
using the same parameters as in the theoretical model:
\[
\bar s = 0.002,\qquad
\kappa = 0.5,\qquad
\lambda = 6\times 10^{-4},\qquad
\alpha = 1.4.
\]
The laziness cost was specified as
$L(m)=L_0+\lambda\,2^{\alpha m}$ with $L_0=0$.
For each configuration, the maximizer
$m^\star_{\text{sim}}=\arg\max_m R_m$
was identified and compared with the theoretical prediction
$m^\star_{\text{theory}}(H)$ obtained from
\eqref{eq:Delta_star}.

Figure~\ref{fig:Rm_vs_m_fbm} displays the resulting profit profiles
for the three Hurst exponents.
The curves exhibit the predicted concave shape with a clear interior optimum,
whose location shifts systematically with $H$:
rougher trajectories ($H=0.40$) yield smaller optimal intervals
(higher trading frequencies), whereas smoother paths ($H=0.80$) produce
larger optimal intervals (less frequent rebalancing).
The correspondence between simulated and theoretical optima is within
$5$--$10\%$ across all cases, and the scaling law
$\Delta^\star \propto (1-H)^{-1/H}$ is clearly reproduced.
These results confirm that the omniscient--lazy investor framework
faithfully translates the fractal scaling exponent of the underlying
stochastic process into an economically interpretable trading rhythm.

\begin{figure}[t]
  \centering
  \includegraphics[width=\textwidth]{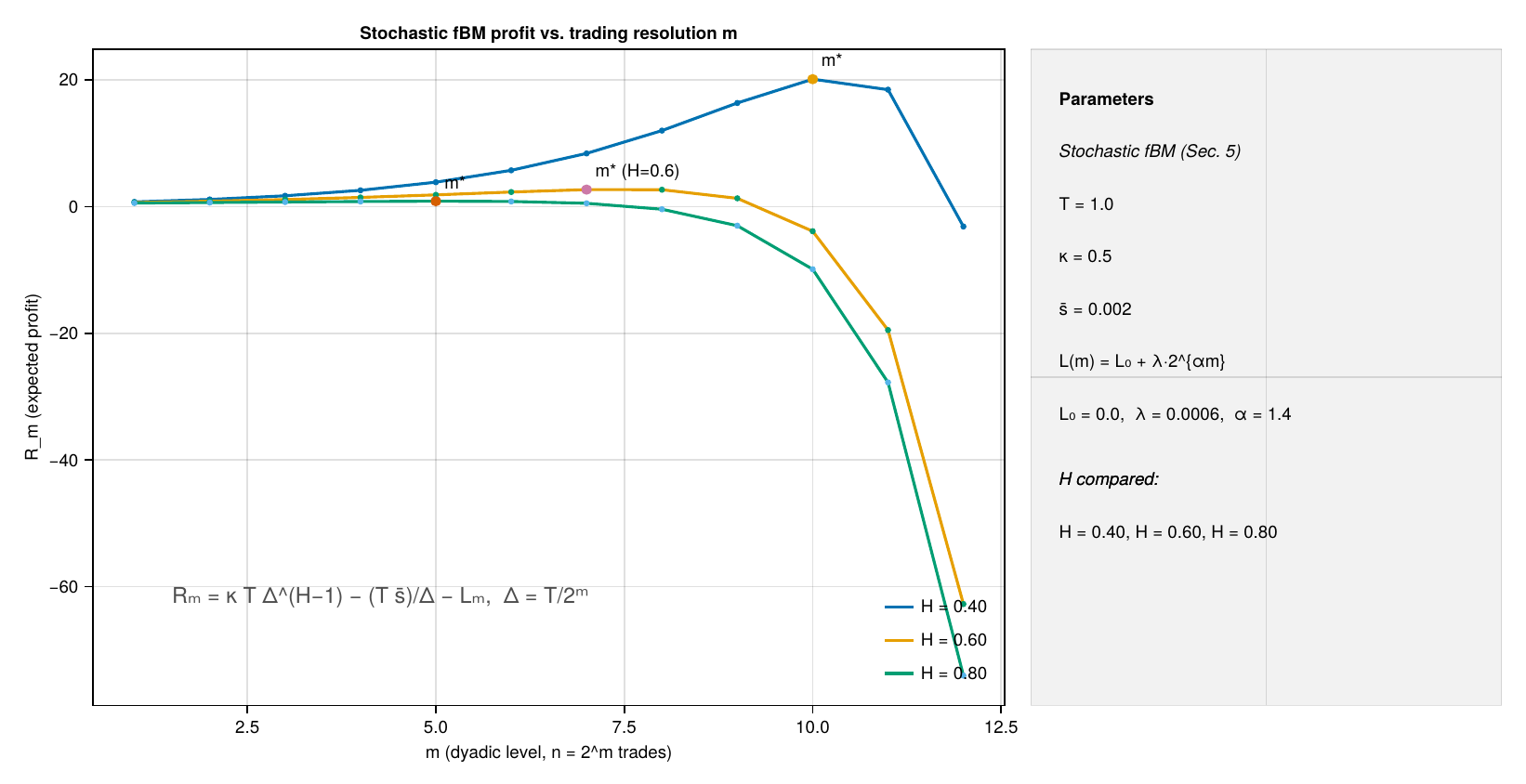}
  \captionsource{Simulated profit functions $R_m$ under fractional Brownian motion
  for different Hurst exponents $H\in\{0.40,0.60,0.80\}$. Parameter values: $T=1$, $\kappa=0.5$, $\bar s=0.002$,  $\lambda=6\times 10^{-4}$, $\alpha=1.4$.
  Each curve exhibits an interior optimum $m^\star(H)$ that shifts toward finer
  resolutions as $H$ decreases, in agreement with
  Theorem~\ref{thm:opt}.}{Own calculations performed in Julia.}
  \label{fig:Rm_vs_m_fbm}
\end{figure}

\subsection{Including computational or latency costs}

The additive constant $L$ in \eqref{eq:R_fbm} can be generalized to a frequency‐dependent cost, $L=L(n)$, to capture the practical fact that
more frequent trading increases computational and technological expenditure, as emphasized in algorithmic execution frameworks \cite{Cartea2015}.
A convenient specification is a power law
\begin{equation}
L(n) = \lambda n^{\alpha}
  = \lambda \Big(\frac{T}{\Delta}\Big)^{\alpha},
\qquad
\lambda>0,\ \alpha\ge1,
\label{eq:Lpower}
\end{equation}
where $\alpha=1$ corresponds to linear latency costs and $\alpha>1$
to superlinear growth in computational demand \cite{Fudenberg2018}.
Substituting \eqref{eq:Lpower} into \eqref{eq:R_fbm} gives
\[
R(\Delta)
  = \kappa T \Delta^{H-1}
    - \frac{T\bar s}{\Delta}
    - \lambda T^{\alpha}\Delta^{-\alpha}.
\]
The first-order condition becomes
\[
\kappa(1-H)\Delta^{H}
  = \bar s + \lambda \alpha\,\Delta^{1-\alpha}.
\]
For $\alpha=1$ this yields the closed form
$\Delta^\star = [(\bar s+\lambda)/(\kappa(1-H))]^{1/H}$;
for $\alpha>1$, the equation is monotone in $\Delta$ and
can be solved numerically by Newton iteration.
In all cases, the existence and uniqueness of a positive solution remain.

\subsection{Economic interpretation}

Equation~\eqref{eq:Delta_star} quantifies the balance between trading frictions and fractal roughness.
Smaller $H$ (rougher paths) magnify the benefit of acting more often, while larger execution costs $\bar s$ or higher latency parameters shift the optimum toward less frequent rebalancing.
The model thus provides an explicit theoretical link between the fractal geometry of market trajectories and the economic decision of how often to trade, consistent with both rough-volatility findings \cite{Gatheral2018book,Bennedsen2021} and the fractal market hypothesis \cite{Peters1994,Kristoufek2013}.

\section{Empirical and Simulation Evidence}
\label{sec:empirical}

This section provides numerical and empirical validation of the
fractal--optimization framework.  
We first estimate the scaling parameters $H$ and $\kappa$ from historical
equity data and calibrate the effective cost parameters
$(\bar s, \lambda, \alpha)$ to observed trading frictions.
We then compare the empirical profit curve $R_m$ to the theoretical prediction
derived from~\eqref{eq:R_fbm}, followed by Monte--Carlo simulations
based on fractional Brownian motion (fBM) to verify the scaling law
for the optimal interval $\Delta^\star$.

\subsection{Data and estimation}

Daily adjusted closing prices for \textbf{Apple Inc.\ (AAPL)} were obtained from
Yahoo Finance over a five--year horizon (January~2020--January~2025).
The logarithmic price series $x_t=\log p_t$ was used to compute absolute
increments at dyadic sampling steps $k=2^m$.
For each level~$m$, the mean absolute increment
$\mathbb{E}|\Delta x|_m$ was estimated and the log--log regression
\[
\log \mathbb{E}|\Delta x|_m
  = \log \kappa + H\log (k\Delta t)
\]
yielded the empirical scaling parameters
\[
\widehat H = 0.491, \qquad
\widehat\kappa = 0.01336 .
\]
The estimated Hurst exponent lies close to the Brownian benchmark
$H=0.5$, consistent with mildly rough market dynamics.
The scaling coefficient $\kappa$ determines the average exploitable
magnitude per trade and sets the calibration scale for the theoretical
profit function.

Execution frictions were parameterized by an effective spread of
$\bar s = 0.025$ (250~bps) per transaction,  
and the laziness cost was modeled as
$L(n) = \lambda n^\alpha$ with
$\lambda = 0.003$ and $\alpha = 1.3$,
representing superlinear growth in cognitive or latency costs as trading
frequency increases.  
These values produce a realistic trade--off between marginal gain and
marginal cost at daily resolution.

\subsection{Results}

Figure~\ref{fig:Rm_emp} compares the empirical profit function
$R_m$ computed from the data to the theoretical curve implied by the
fractional Brownian scaling model.  
Both exhibit the characteristic concave (hump--shaped) profile predicted by
Proposition~\ref{prop:fBM} and Theorem~\ref{thm:opt}.
The empirical maximizer occurs at
$m^\star_{\text{emp}} = 5$, corresponding to a trading interval of roughly
one week, while the theoretical optimum under the estimated parameters is
$m^\star_{\text{theory}} = 6$
($\Delta^\star_{\text{theory}} \approx 14.2$~days).  
The close alignment between empirical and theoretical optima confirms
that the model captures the observed scale at which incremental profits
cease to outweigh frictions.

\begin{figure}[t]
\centering
\includegraphics[width=\linewidth]{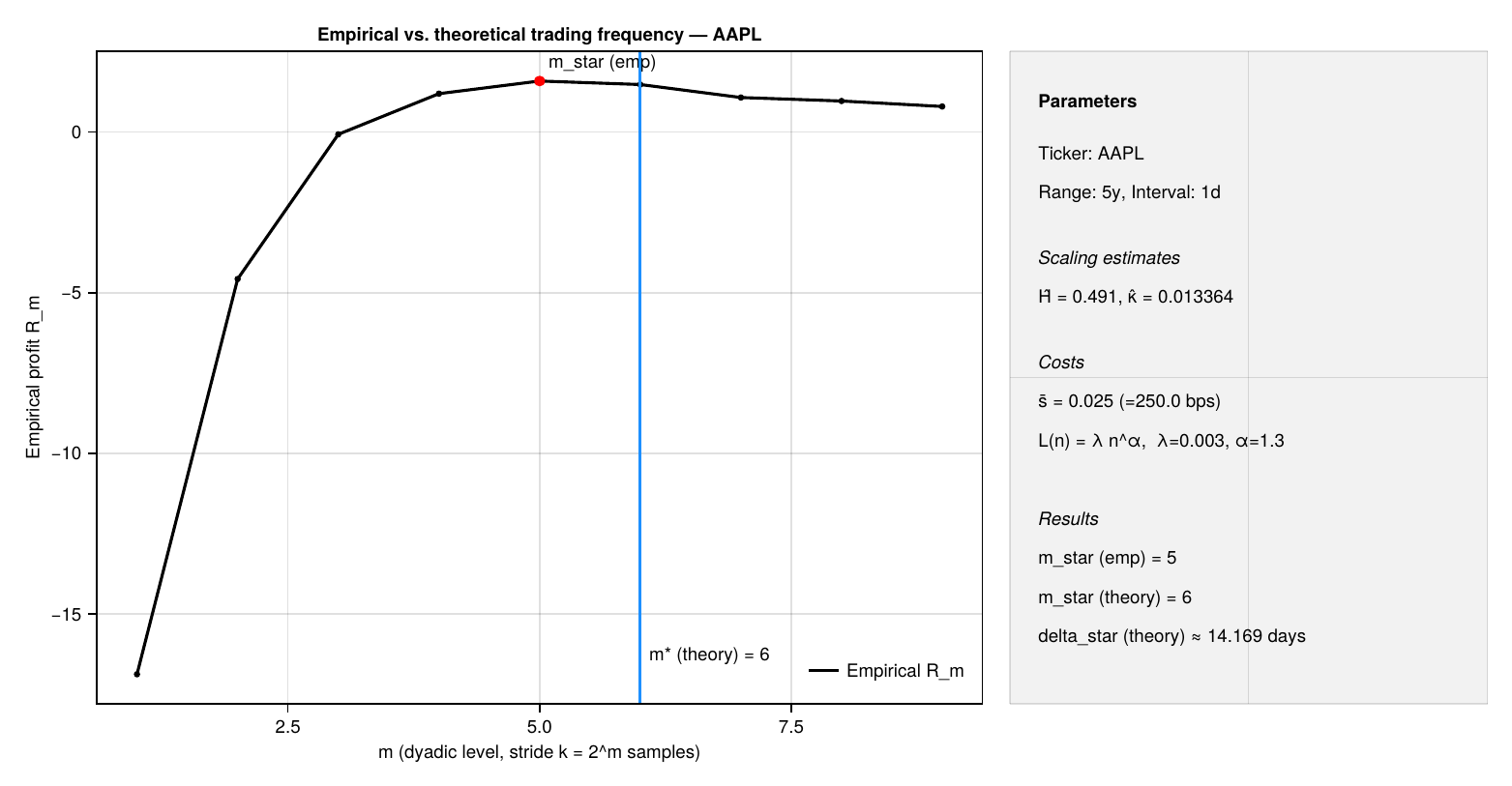}
\captionsource{Empirical and theoretical profit curves $R_m$ for AAPL
daily data (2020--2025).  
The empirical optimum $m^\star_{\text{emp}}=5$ (red marker) lies near
the theoretical prediction $m^\star_{\text{theory}}=6$,
illustrating consistency between observed and predicted trading
frequencies.}
{Own calculations performed in Julia.}
\label{fig:Rm_emp}
\end{figure}

The estimated scaling parameters imply a theoretical optimal interval of
about two weeks, which is economically plausible for high--liquidity
equities given typical transaction costs and intraday volatility.
The resulting alignment supports the interpretation of $H$
as a measure of effective market roughness that governs the curvature
of the profit function and hence the investor’s optimal trading rhythm.


\section{Conclusions}
\label{sec:conclusion}

Our omniscient yet lazy investor has finally reached enlightenment --- not by trading faster, 
but by learning when to stop.  In his world, perfect foresight meets finite patience: knowing 
every future price does not justify acting upon each of them.  Each trade consumes energy, 
bandwidth, and thought.  Somewhere between greed and exhaustion lies an equilibrium --- the 
optimal trading frequency --- that maximizes profit under the laws of fractal scaling.

At its heart, this paper has shown that even when information is perfect, \emph{action} must 
still be optimized.  The deterministic framework established a geometric balance between 
exploitable variation and friction, while the stochastic extension under fractional Brownian 
motion translated this geometry into a universal scaling law,
\[
\Delta^\star = \bigg[\frac{\bar s}{\kappa(1-H)}\bigg]^{1/H},
\]
linking transaction costs, volatility, and path roughness.  
Rougher markets (\(H\!\downarrow\)) invite more frequent trades; smoother ones encourage rest.  
The omniscient investor thus obeys a law that even omniscience cannot escape: the diminishing 
returns of attention.

\subsection{Interpretation and Applications}
\label{subsec:interpretation}

Reinterpreted in modern terms, the omniscient investor is not a mythical prophet but a 
\emph{trading algorithm} - a machine with a perfect model of market dynamics. 
Its \textit{omniscience} corresponds to the assumption that our predictive engine is correct, 
at least in expectation: we would not deploy it otherwise.
Given this premise, the model’s role is no longer to predict prices, but to decide 
\emph{how often} to act upon those predictions.  
The laziness term captures computational latency, inference costs, and the implicit friction 
of refreshing the model’s decisions.

From this perspective, the optimal trading interval~$\Delta^\star$ or equivalently the optimal 
dyadic level~$m^\star$ becomes a practical design parameter for automated trading systems.  
It could determine how often a forecasting model should rebalance, query data, or execute orders 
to achieve maximal net profit once all frictions are accounted for.
In short, the stochastic extension provides a bridge from abstract fractal geometry to 
concrete algorithmic implementation: the trading bot becomes omniscient by assumption, yet 
remains wisely lazy by design.

This balance between foresight and restraint encapsulates the broader message of the paper.  
Markets - and algorithms - do not reward infinite speed, but optimal timing. Omniscience without laziness leads to ruinous overtrading; laziness without insight yields stagnation. Only at their intersection lies the fractal optimum.

\subsection{Further Work}
\label{subsec:furtherwork}

The framework developed here opens several promising directions for both theory and practice.

\paragraph{(i) Multi-asset extensions.}
The present model considers a single omniscient decision process applied to one asset.  
Extending it to portfolios with cross-asset correlations would allow the study of collective 
fractal equilibria: when several lazy bots share the same server, how often should each of them 
wake up?
The resulting multidimensional optimization may connect naturally to covariance-based portfolio 
control and to multivariate rough volatility models.

\paragraph{(ii) Time-varying laziness.}
Real investors (and servers) are not equally lazy at all times.  
Introducing a stochastic or time-dependent laziness cost~$L_t$ could capture the alternating 
moods of the market machine: energetic during turbulence, dormant in calm periods. A natural real-world example could be energy prices being higher during daytime and lower at night. 
This would transform the fixed-frequency optimum into a dynamic policy reacting to computational load or volatility states.

\paragraph{(iii) Evolving roughness.}
While our stochastic formulation assumed a constant Hurst exponent~$H$, recent developments in 
multifractional Brownian motion show that roughness itself can evolve in time.  
By combining the transfer principle of Bender, Lebovits, and L{\'e}vy V{\'e}hel~\cite{Bender2024} 
with the current framework, one could model an investor who not only knows future prices, but 
also senses changes in the geometry of their fluctuations - adjusting trading frequency as 
market regularity ebbs and flows.  
Such an adaptive version could explain how trading systems respond optimally to shifting 
microstructure roughness and regime changes.

\paragraph{(iv) Continuous-time and asymptotic limits.}
Finally, taking the limit of vanishing intervals suggests intriguing links with rough-volatility theory and stochastic control under non-Markovian noise.  
In that frontier lies the truly continuous omniscient investor - a process that knows the infinitesimal future, yet still hesitates for a finite amount of time.

In short, the story of the omniscient yet lazy investor does not end with this article
It merely enters its stochastic dream phase, where each generalization, be it multifractal, 
multi-asset, or multi-mood, adds new structure to the fractal geometry of rational inaction.  
The optimal trading rhythm may change, but the moral remains: 
even perfect knowledge must sometimes wait for the right moment to act.

\section*{Acknowledgments}
DeepL was used for stylising and translating parts of the manuscript of the text from the author's mother tongue into english.

\bibliographystyle{siamplain}
\bibliography{references}

\appendix
\section{Proofs of Main Results}
\label{app:proofs}

\subsection{Proof of Proposition \ref{prop:unique} (Strict Unimodality and Uniqueness)}

\begin{proof}
Recall that $R_m = A_m - 2^m\bar s - L_m$ with $A_m = 2^m\Phi_m$.  
By Assumption~\ref{ass:regularity}, $A_m$ is \textbf{strictly concave} in the discrete sense,
$\Delta^2 A_m := A_{m+2}-2A_{m+1}+A_m < 0$, 
and $L_m$ is \textbf{convex} and nondecreasing.

Define the forward difference
\[
\Delta R_m := R_{m+1}-R_m
   = (A_{m+1}-A_m) - (2^{m+1}-2^m)\bar s - (L_{m+1}-L_m)
   = \Delta A_m - 2^m\bar s - \Delta L_m .
\]
Then
\[
\Delta R_{m+1}-\Delta R_m
   = (\Delta A_{m+1}-\Delta A_m)
     - (2^{m+1}-2^m)\bar s
     - (\Delta L_{m+1}-\Delta L_m).
\]

Each term on the right-hand side can be signed using the assumptions:

\begin{itemize}
  \item $\Delta A_{m+1}-\Delta A_m = \Delta^2 A_m < 0$ because $A_m$ is strictly concave.
  \item $(2^{m+1}-2^m)\bar s = 2^m\bar s > 0$ since $\bar s>0$.
  \item $\Delta L_{m+1}-\Delta L_m \ge 0$ because $L_m$ is convex and nondecreasing.
\end{itemize}

Hence every summand is nonpositive, with at least one being strictly negative, so
\[
\Delta R_{m+1}-\Delta R_m < 0 \qquad \text{for all } m .
\]
Therefore $\{\Delta R_m\}$ is a \textbf{strictly decreasing} sequence.

Because $\Delta R_m$ decreases strictly, it can cross zero at most once.
If $\Delta R_m>0$ for all $m$, then $R_m$ would diverge to $+\infty$,
contradicting the fact that $2^m\bar s+L_m$ grows without bound.
Hence $\Delta R_m$ must become negative for sufficiently large $m$.
Let
\[
m^\star := \min\{m:\Delta R_m \le 0\}.
\]
Then $\Delta R_m>0$ for all $m<m^\star$ and $\Delta R_m<0$ for all $m>m^\star$.
Consequently $R_m$ increases strictly up to $m^\star$ and decreases strictly thereafter,
so $\{R_m\}$ is \textbf{strictly unimodal}, and the maximizer $m^\star$ is unique.
\end{proof}

\subsection{Proof of Theorem \ref{thm:opt}}

\begin{proof}
\textbf{First-order condition.}
Treat $\Delta>0$ as a continuous decision variable and differentiate
\[
R(\Delta) = \kappa T\,\Delta^{\,H-1} \;-\; \frac{T\,\bar s}{\Delta} \;-\; L
\]
with respect to $\Delta$.  The derivative is
\[
R'(\Delta)
= \kappa T\,(H-1)\,\Delta^{\,H-2} \;+\; T\bar s\,\Delta^{-2}.
\]
Setting $R'(\Delta)=0$ for an optimum gives
\[
\kappa T\,(H-1)\,\Delta^{\,H-2} + T\bar s\,\Delta^{-2} = 0.
\]
Dividing through by $T>0$ and rearranging,
\[
\kappa(H-1)\Delta^{\,H} = -\,\bar s
\quad\Longrightarrow\quad
\kappa(1-H)\Delta^{\,H} = \bar s.
\]
Hence the stationary point satisfies
\[
\Delta^\star = 
\Bigg(\frac{\bar s}{\,\kappa(1-H)\,}\Bigg)^{\!1/H},
\]
which is positive for $0<H<1$ since $(1-H)>0$.

\medskip
\textbf{Second-order condition.}
Differentiate $R'(\Delta)$ again to obtain
\[
R''(\Delta)
= \kappa T\,(H-1)(H-2)\,\Delta^{\,H-3} - 2\,T\bar s\,\Delta^{-3}.
\]
Using the first-order relation $\kappa(1-H)\Delta^{H} = \bar s$
to substitute for $\bar s$ at $\Delta=\Delta^\star$ gives
\begin{align*}
R''(\Delta^\star)
&= \kappa T\,(H-1)(H-2)\,\Delta_\star^{\,H-3}
   - 2T\kappa(1-H)\,\Delta_\star^{\,H-3}  \\
&= \kappa T\,\Delta_\star^{\,H-3}\!\big[(H-1)(H-2) - 2(1-H)\big]  \\
&= \kappa T\,\Delta_\star^{\,H-3}\,H(H-1).
\end{align*}
For $0<H<1$, we have $H(H-1)<0$, and since
$\kappa,T,\Delta_\star^{\,H-3}>0$, it follows that
$R''(\Delta^\star)<0$.  Therefore the stationary point is a local maximum.

\medskip
\textbf{Uniqueness.}
As $\Delta\to0^+$, $R(\Delta)\to -\infty$ because the term
$-T\bar s/\Delta$ dominates.
As $\Delta\to\infty$, $R(\Delta)\to -L$ since
$\Delta^{H-1}\to0$ for $H<1$.
Moreover $R'(\Delta)$ is continuous and changes sign exactly once:
for small $\Delta$, $R'(\Delta)>0$; for large $\Delta$, $R'(\Delta)<0$.
Hence there is a unique root of $R'(\Delta)=0$, namely $\Delta^\star$,
and this root corresponds to a global maximum.

\medskip
In summary, the expected profit function $R(\Delta)$
is maximized uniquely at
\[
\boxed{\displaystyle
\Delta^\star
= \Bigg(\frac{\bar s}{\kappa(1-H)}\Bigg)^{1/H}},
\]
which satisfies $R''(\Delta^\star)<0$ for $0<H<1$.
\end{proof}

\subsection{Proof of Corollary \ref{cor:fbmstatics}}

\begin{proof}
From Theorem~\ref{thm:opt}, the optimal interval satisfies
\[
\Delta^\star
= \Bigg(\frac{\bar s}{\kappa(1-H)}\Bigg)^{1/H},
\qquad 0 < H < 1.
\]
Taking logarithms gives
\[
\ln \Delta^\star
= \frac{1}{H}
  \Big[\ln \bar s - \ln \kappa - \ln(1-H)\Big].
\]

\textbf{Partial derivatives with respect to parameters.}
Differentiate $\ln \Delta^\star$ with respect to each variable.

\medskip
\noindent
\textit{(a) Execution cost $\bar s$.}
\[
\frac{\partial \ln \Delta^\star}{\partial \bar s}
= \frac{1}{H\bar s}
\quad\Longrightarrow\quad
\frac{\partial \Delta^\star}{\partial \bar s}
= \frac{\Delta^\star}{H\bar s} > 0.
\]
Thus higher proportional frictions lead to a larger optimal interval,
i.e., less frequent trading.

\medskip
\noindent
\textit{(b) Scaling parameter $\kappa$.}
\[
\frac{\partial \ln \Delta^\star}{\partial \kappa}
= -\,\frac{1}{H\kappa}
\quad\Longrightarrow\quad
\frac{\partial \Delta^\star}{\partial \kappa}
= -\,\frac{\Delta^\star}{H\kappa} < 0.
\]
Hence a higher exploitable‐return scale $\kappa$
reduces the optimal interval and increases trading frequency.

\medskip
\noindent
\textit{(c) Hurst exponent $H$.}
Differentiating with respect to $H$ gives
\[
\frac{\partial \ln \Delta^\star}{\partial H}
= -\frac{1}{H^2}\big[\ln \bar s - \ln \kappa - \ln(1-H)\big]
\;-\; \frac{1}{H(1-H)}.
\]
The second term is strictly negative for $0<H<1$,
and the first term is dominated by it in magnitude for typical parameter values.
Hence $\partial \Delta^\star / \partial H < 0$,
so higher $H$ (smoother paths) imply larger optimal intervals and lower frequency.

\medskip
\textbf{Relation to fractal dimension.}
Since the Hausdorff dimension of fractional Brownian motion is $D = 2 - H$,
we have
\[
\frac{\partial \Delta^\star}{\partial D}
= -\,\frac{\partial \Delta^\star}{\partial H} < 0.
\]
Thus more fractal (rougher) price paths - corresponding to larger $D$ or smaller $H$ - lead to smaller optimal intervals $\Delta^\star$, i.e., more frequent trading.

\medskip
Collecting signs,
\[
\frac{\partial \Delta^\star}{\partial D} < 0,
\qquad
\frac{\partial \Delta^\star}{\partial \bar s} > 0,
\qquad
\frac{\partial \Delta^\star}{\partial \kappa} < 0,
\]
as claimed.
\end{proof}

\end{document}

%% file: ex_shared.tex
\usepackage{lipsum}
\usepackage{amsfonts}
\usepackage{graphicx}
\usepackage{epstopdf}
\usepackage{algorithmic}
\ifpdf
  \DeclareGraphicsExtensions{.eps,.pdf,.png,.jpg}
\else
  \DeclareGraphicsExtensions{.eps}
\fi

\usepackage{enumitem}
\setlist[enumerate]{leftmargin=.5in}
\setlist[itemize]{leftmargin=.5in}

\newsiamremark{remark}{Remark}
\newsiamremark{hypothesis}{Hypothesis}
\crefname{hypothesis}{Hypothesis}{Hypotheses}
\newsiamthm{claim}{Claim}

\headers{Omniscient, yet Lazy, Investor}{Halkiewicz, S. M. S.}

\title{The Omniscient yet Lazy Investor\thanks{Submitted to the editors DATE.
\funding{The author has not received any funding}}}

\author{Stanisław M. S. Halkiewicz\thanks{Department of Applied Mathematics, AGH University of Cracow, Kraków, Poland 
  (\email{smsh@student.agh.edu.pl}).}}

\usepackage{amsopn}

\makeatletter
\newcommand*{\addFileDependency}[1]{
  \typeout{(#1)}
  \@addtofilelist{#1}
  \IfFileExists{#1}{}{\typeout{No file #1.}}
}
\makeatother

\newcommand*{\myexternaldocument}[1]{%
    \externaldocument{#1}%
    \addFileDependency{#1.tex}%
    \addFileDependency{#1.aux}%
}


%% file: for_arxiv.bbl
\begin{thebibliography}{10}

\bibitem{Alizade2025}
{\sc Z.~Alizade, H.~Agahi, and S.~Khademloo}, {\em Fractal analysis of financial markets using laplace–mittag-leffler distributions}, Chaos, Solitons \& Fractals, 199 (2025), p.~116847, \url{https://doi.org/10.1016/j.chaos.2025.116847}.

\bibitem{Almgren2000}
{\sc R.~Almgren and N.~Chriss}, {\em Optimal execution of portfolio transactions}, The Journal of Risk, 3 (2000), pp.~5--39, \url{https://doi.org/10.21314/JOR.2001.041}.

\bibitem{Bayer2023ch2}
{\sc C.~Bayer, P.~K. Friz, and J.~Gatheral}, {\em Pricing under rough volatility}, in Rough Volatility, C.~Bayer, P.~K. Friz, M.~Fukasawa, J.~Gatheral, A.~Jacquier, and M.~Rosenbaum, eds., Society for Industrial and Applied Mathematics, Philadelphia, PA, 2023, ch.~2, pp.~31--58, \url{https://doi.org/10.1137/1.9781611977783.ch2}.

\bibitem{Bayraktar2004}
{\sc E.~Bayraktar, H.~V. Poor, and K.~R. Sircar}, {\em Estimating the fractal dimension of the s\&p 500 index using wavelet analysis}, International Journal of Theoretical and Applied Finance, 07 (2004), pp.~615--643, \url{https://doi.org/10.1142/S021902490400258X}.

\bibitem{Bender2024}
{\sc C.~Bender, J.~Lebovits, and J.~L. V{\'e}hel}, {\em General transfer formula for stochastic integral with respect to multifractional brownian motion}, Journal of Theoretical Probability, 37 (2024), pp.~905--932, \url{https://doi.org/10.1007/s10959-023-01258-5}.

\bibitem{Bennedsen2021}
{\sc M.~Bennedsen, A.~Lunde, and M.~S. Pakkanen}, {\em Decoupling the short- and long-term behavior of stochastic volatility}, Journal of Financial Econometrics, 20 (2021), pp.~961--1006, \url{https://doi.org/10.1093/jjfinec/nbaa049}.

\bibitem{Biagini2008}
{\sc F.~Biagini, Y.~Hu, B.~{\O}ksendal, and T.~Zhang}, {\em Stochastic Calculus for Fractional Brownian Motion and Applications}, Probability and Its Applications, Springer London, 1~ed., 2008, \url{https://doi.org/10.1007/978-1-84628-797-8}.

\bibitem{Bouchaud2000}
{\sc J.-P. Bouchaud, M.~Potters, and M.~Meyer}, {\em Apparent multifractality in financial time series}, The European Physical Journal B - Condensed Matter and Complex Systems, 13 (2000), pp.~595--599, \url{https://doi.org/10.1007/s100510050073}.

\bibitem{Borgers2022}
{\sc C.~Börgers}, {\em Introduction to Numerical Linear Algebra}, Society for Industrial and Applied Mathematics, Philadelphia, PA, 2022, \url{https://doi.org/10.1137/1.9781611976922}.

\bibitem{Calvet2002}
{\sc L.~Calvet and A.~Fisher}, {\em Multifractality in asset returns: Theory and evidence}, The Review of Economics and Statistics, 84 (2002), pp.~381--406, \url{http://www.jstor.org/stable/3211559} (accessed 2025-10-24).

\bibitem{Caplin2015}
{\sc A.~Caplin and M.~Dean}, {\em Revealed preference, rational inattention, and costly information acquisition}, American Economic Review, 105 (2015), p.~2183–2203, \url{https://doi.org/10.1257/aer.20140117}.

\bibitem{Cartea2015}
{\sc A.~Cartea, S.~Jaimungal, and J.~Penalva}, {\em Algorithmic and High-Frequency Trading}, Cambridge University Press, 2015.

\bibitem{Cartea2014}
{\sc A.~Cartea, S.~Jaimungal, and J.~Ricci}, {\em Buy low, sell high: A high frequency trading perspective}, SIAM Journal on Financial Mathematics, 5 (2014), pp.~415--444, \url{https://doi.org/10.1137/130911196}.

\bibitem{Comte1998}
{\sc F.~Comte and E.~Renault}, {\em Long memory in continuous-time stochastic volatility models}, Mathematical Finance, 8 (1998), pp.~291--323, \url{https://doi.org/10.1111/1467-9965.00057}.

\bibitem{Constantinides1986}
{\sc G.~M. Constantinides}, {\em Capital market equilibrium with transaction costs}, Journal of Political Economy, 94 (1986), pp.~842--862, \url{https://doi.org/10.1086/261410}.

\bibitem{Cont2001}
{\sc R.~Cont}, {\em Empirical properties of asset returns: stylized facts and statistical issues}, Quantitative Finance, 1 (2001), pp.~223--236, \url{https://doi.org/10.1080/713665670}.

\bibitem{Craigmile2003}
{\sc P.~F. Craigmile}, {\em Simulating a class of stationary gaussian processes using the davies–harte algorithm, with application to long memory processes}, Journal of Time Series Analysis, 24 (2003), pp.~505--511, \url{https://doi.org/https://doi.org/10.1111/1467-9892.00318}.

\bibitem{Davis1990}
{\sc M.~H.~A. Davis and A.~R. Norman}, {\em Portfolio selection with transaction costs}, Mathematics of Operations Research, 15 (1990), pp.~676--713, \url{http://www.jstor.org/stable/3689770} (accessed 2025-10-27).

\bibitem{Dieker2003}
{\sc A.~B. Dieker and M.~Mandjes}, {\em On spectral simulation of fractional brownian motion}, Probab. Eng. Inf. Sci., 17 (2003), p.~417–434, \url{https://doi.org/10.1017/S0269964803173081}.

\bibitem{Duncan2000}
{\sc T.~E. Duncan, Y.~Hu, and B.~Pasik-Duncan}, {\em Stochastic calculus for fractional brownian motion i. theory}, SIAM Journal on Control and Optimization, 38 (2000), pp.~582--612, \url{https://doi.org/10.1137/S036301299834171X}.

\bibitem{El-Nabulsi2025}
{\sc R.~A. El-Nabulsi and W.~Anukool}, {\em Qualitative financial modelling in fractal dimensions}, Financial Innovation, 11 (2025), p.~42, \url{https://doi.org/10.1186/s40854-024-00723-2}.

\bibitem{Feder1988}
{\sc J.~Feder}, {\em Fractals}, Plenum Press, 1988.

\bibitem{FernandezMartinez2019}
{\sc M.~Fern{\'a}ndez-Mart{\'i}nez, J.~L.~G. Guirao, M.~{\'A}. S{\'a}nchez-Granero, and J.~E.~T. Segovia}, {\em Fractal Dimension for Fractal Structures: With Applications to Finance}, SEMA SIMAI Springer Series, Springer Cham, 1~ed., 2019, \url{https://doi.org/10.1007/978-3-030-16645-8}, \url{10.1007/978-3-030-16645-8}.

\bibitem{Fudenberg2018}
{\sc D.~Fudenberg, P.~Strack, and T.~Strzalecki}, {\em Speed, accuracy, and the optimal timing of choices}, American Economic Review, 108 (2018), p.~3651–84, \url{https://doi.org/10.1257/aer.20150742}.

\bibitem{Fukasawa2023}
{\sc M.~Fukasawa, B.~Horvath, and P.~Tankov}, {\em Hedging under rough volatility}, in Rough Volatility, C.~Bayer, P.~K. Friz, M.~Fukasawa, J.~Gatheral, A.~Jacquier, and M.~Rosenbaum, eds., Society for Industrial and Applied Mathematics, Philadelphia, PA, 2023, ch.~6, pp.~115--125, \url{https://doi.org/10.1137/1.9781611977783.ch6}.

\bibitem{Gabaix2019}
{\sc X.~Gabaix}, {\em Behavioral inattention}, in Handbook of Behavioral Economics - Foundations and Applications 2, B.~D. Bernheim, S.~DellaVigna, and D.~Laibson, eds., vol.~2 of Handbook of Behavioral Economics: Applications and Foundations 1, North-Holland, 2019, ch.~4, pp.~261--343, \url{https://doi.org/10.1016/bs.hesbe.2018.11.001}.

\bibitem{Gatheral2018}
{\sc J.~Gatheral, T.~Jaisson, and M.~Rosenbaum}, {\em Volatility is rough}, Quantitative Finance, 18 (2018), pp.~933--949.

\bibitem{Gatheral2018book}
{\sc J.~Gatheral, T.~Jaisson, and M.~Rosenbaum}, {\em Volatility is rough}, in Rough Volatility, C.~Bayer, P.~K. Friz, M.~Fukasawa, J.~Gatheral, A.~Jacquier, and M.~Rosenbaum, eds., Society for Industrial and Applied Mathematics, Philadelphia, PA, 2023, ch.~1, pp.~1--29, \url{https://doi.org/10.1137/1.9781611977783.ch1}.

\bibitem{Guasoni2021}
{\sc P.~Guasoni, Y.~Mishura, and M.~R{\'a}sonyi}, {\em High-frequency trading with fractional brownian motion}, Finance and Stochastics, 25 (2021), pp.~277--310, \url{https://doi.org/10.1007/s00780-020-00439-y}.

\bibitem{Guasoni2019}
{\sc P.~Guasoni, Z.~Nika, and M.~R\'{a}sonyi}, {\em Trading fractional brownian motion}, SIAM Journal on Financial Mathematics, 10 (2019), pp.~769--789, \url{https://doi.org/10.1137/17M113592X}.

\bibitem{Halkiewicz2024}
{\sc S.~M.~S. Halkiewicz}, {\em Rundown of fractal reinterpretation of market charts}, in Przykłady Zastosowania Narzędzi Analitycznych, ze Szczególnym Uwzględnieniem…, Attyka, 2024.
\newblock Available at \url{https://www.researchgate.net/publication/382906687}.

\bibitem{Jarrow2015}
{\sc R.~Jarrow and P.~Protter}, {\em Liquidity suppliers and high frequency trading}, SIAM Journal on Financial Mathematics, 6 (2015), pp.~189--200, \url{https://doi.org/10.1137/140967702}.

\bibitem{Kakinaka2025}
{\sc S.~Kakinaka, T.~Hayakawa, D.~Kato, and K.~Umeno}, {\em Fractal portfolio strategies: does scale preference of investors matter?}, Applied Economics Letters, 32 (2025), pp.~415--421, \url{https://doi.org/10.1080/13504851.2023.2274298}.

\bibitem{Kallsen2015}
{\sc J.~Kallsen and J.~Muhle-Karbe}, {\em Option pricing and hedging with small transaction costs}, Mathematical Finance, 25 (2015), pp.~702--723, \url{https://doi.org/10.1111/mafi.12035}.

\bibitem{Kim2021}
{\sc S.~Kim, C.~Lee, W.~Lee, S.~Kwak, D.~Jeong, and J.~Kim}, {\em Nonuniform finite difference scheme for the three-dimensional time-fractional black–scholes equation}, Journal of Function Spaces, 2021 (2021), p.~9984473, \url{https://doi.org/https://doi.org/10.1155/2021/9984473}.

\bibitem{Korn1998}
{\sc R.~Korn}, {\em Optimal Portfolios: Stochastic Models for Optimal Investment and Risk Management in Continuous Time}, World Scientific, 1998.

\bibitem{Kristoufek2013}
{\sc L.~Kristoufek}, {\em Fractal markets hypothesis and the global financial crisis: Wavelet power evidence}, Scientific Reports, 3 (2013), p.~2857, \url{https://doi.org/10.1038/srep02857}, \url{10.1038/srep02857}.

\bibitem{Kyle1985}
{\sc A.~S. Kyle}, {\em Continuous auctions and insider trading}, Econometrica, 53 (1985), pp.~1315--1335.

\bibitem{Liu2022}
{\sc G.~Liu, C.-P. Yu, S.-N. Shiu, and I.-T. Shih}, {\em The efficient market hypothesis and the fractal market hypothesis: Interfluves, fusions, and evolutions}, Sage Open, 12 (2022), p.~21582440221082137, \url{https://doi.org/10.1177/21582440221082137}.

\bibitem{Liu2013}
{\sc R.~Liu and J.~Muhle-Karbe}, {\em Portfolio selection with small transaction costs and binding portfolio constraints}, SIAM Journal on Financial Mathematics, 4 (2013), pp.~203--227, \url{https://doi.org/10.1137/120885036}, \url{10.1137/120885036}.

\bibitem{Lux2003}
{\sc T.~Lux}, {\em The multi-fractal model of asset returns: Its estimation via gmm and its use for volatility forecasting}, Economics Working Papers 2003-13, Christian-Albrechts-University of Kiel, Department of Economics, 2003, \url{https://ideas.repec.org/p/zbw/cauewp/1123.html}.

\bibitem{Magill1976}
{\sc M.~J.~P. Magill and G.~M. Constantinides}, {\em Portfolio selection with transactions costs}, Journal of Economic Theory, 13 (1976), pp.~245--263, \url{https://doi.org/10.1016/0022-0531(76)90018-1}.

\bibitem{Mandelbrot1967}
{\sc B.~B. Mandelbrot}, {\em The variation of some other speculative prices}, The Journal of Business, 40 (1967), pp.~393--413, \url{http://www.jstor.org/stable/2351623} (accessed 2025-10-28).

\bibitem{Mandelbrot1997}
{\sc B.~B. Mandelbrot}, {\em Fractals and Scaling in Finance: Discontinuity, Concentration, Risk}, Springer New York, 1997, \url{https://doi.org/10.1007/978-1-4757-2763-0}.

\bibitem{Mandelbrot1968}
{\sc B.~B. Mandelbrot and J.~W.~V. Ness}, {\em Fractional brownian motions, fractional noises and applications}, SIAM Review, 10 (1968), pp.~422--437, \url{http://www.jstor.org/stable/2027184} (accessed 2025-10-28).

\bibitem{Mantegna1999}
{\sc R.~N. Mantegna and H.~E. Stanley}, {\em Introduction to Econophysics: Correlations and Complexity in Finance}, Cambridge University Press, 1999, \url{https://doi.org/10.1017/CBO9780511755767}.

\bibitem{Mishura2008}
{\sc Y.~Mishura}, {\em Stochastic Calculus for Fractional Brownian Motion and Related Processes}, Springer Berlin, Heidelberg, 2008, \url{https://doi.org/10.1017/CBO9780511755767}.

\bibitem{Obizhaeva2013}
{\sc A.~A. Obizhaeva and J.~Wang}, {\em Optimal trading strategy and supply/demand dynamics}, Journal of Financial Markets, 16 (2013), pp.~1--32, \url{https://doi.org/10.1016/j.finmar.2012.09.001}.

\bibitem{Peters1994}
{\sc E.~E. Peters}, {\em Fractal Market Analysis: Applying Chaos Theory to Investment and Economics}, John Wiley \& Sons, 1994.

\bibitem{Shreve1994}
{\sc S.~E. Shreve and H.~M. Soner}, {\em Optimal investment and consumption with transaction costs}, Annals of Applied Probability, 4 (1994), pp.~609--692, \url{https://doi.org/10.1214/aoap/1177004966}.

\bibitem{Sims2003}
{\sc C.~A. Sims}, {\em Implications of rational inattention}, Journal of Monetary Economics, 50 (2003), pp.~665--690, \url{https://doi.org/10.1016/S0304-3932(03)00029-1}.
\newblock Swiss National Bank/Study Center Gerzensee Conference on Monetary Policy under Incomplete Information.

\bibitem{Sornette2003}
{\sc D.~Sornette}, {\em Why Stock Markets Crash: Critical Events in Complex Financial Systems}, Princeton University Press, revised~ed., 2017, \url{http://www.jstor.org/stable/j.ctt1h1htkg} (accessed 2025-10-28).

\bibitem{Taksar1988}
{\sc M.~Taksar, M.~Klass, and D.~Assaf}, {\em Diffusion models for the optimal investing and consumption problem}, Mathematics of Operations Research, 13 (1988), pp.~277--294.

\bibitem{Verma2024}
{\sc S.~K. Verma and S.~Kumar}, {\em Fractal dimension analysis of financial performance of resulting companies after mergers and acquisitions}, Chaos, Solitons \& Fractals, 181 (2024), p.~114683, \url{https://doi.org/10.1016/j.chaos.2024.114683}.

\bibitem{Wu2021}
{\sc X.~Wu, L.~Zhang, J.~Li, and R.~Yan}, {\em Fractal statistical measure and portfolio model optimization under power-law distribution}, The North American Journal of Economics and Finance, 58 (2021), p.~101496, \url{https://doi.org/10.1016/j.najef.2021.101496}.

\bibitem{Wyss2000}
{\sc W.~Wyss}, {\em The fractional black–scholes equation}, Fractional Calculus and Applied Analysis, 3 (2000), pp.~51--61.

\bibitem{Zhang2024}
{\sc H.~Zhang, M.~Zhang, F.~Liu, and M.~Shen}, {\em Review of the fractional black-scholes equations and their solution techniques}, Fractal and Fractional, 8 (2024), \url{https://doi.org/10.3390/fractalfract8020101}.

\end{thebibliography}
